\documentclass{article}[12pt]
\usepackage{amsmath,amsthm,amssymb,amsfonts,amscd,eucal}

\numberwithin{equation}{section}

\DeclareMathOperator{\re}{Re}
\DeclareMathOperator{\Var}{Var} 
\DeclareMathOperator{\Cov}{Cov}
\DeclareMathOperator{\Corr}{Corr} 
\DeclareMathOperator{\Tr}{Tr}

\newtheorem{Thm}{Theorem}[section]

\newtheorem{Lemma}[Thm]{Lemma}
\theoremstyle{definition}
\newtheorem{Dfn}[Thm]{Definition}
\theoremstyle{remark}
\newtheorem{Rem}[Thm]{Remark}

\begin{document}

\title{An inequality related to uncertainty principle in von Neumann
algebras}

\author{ Paolo Gibilisco\footnote{Dipartimento SEFEMEQ, Facolt\`a di
Economia, Universit\`a di Roma ``Tor Vergata", Via Columbia 2, 00133
Rome, Italy.  Email: gibilisco@volterra.uniroma2.it -- URL:
http://www.economia.uniroma2.it/sefemeq/professori/gibilisco} \ and
Tommaso Isola\footnote{Dipartimento di Matematica, Universit\`a di
Roma ``Tor Vergata", Via della Ricerca Scientifica, 00133 Rome, Italy. 
Email: isola@mat.uniroma2.it -- URL:
http://www.mat.uniroma2.it/$\sim$isola} }

\maketitle

\begin{abstract}
Recently Kosaki proved in \cite{Ko} an inequality for matrices that
can be seen as a kind of new uncertainty principle.  Independently,
the same result was proved by Yanagi {\sl et al.} in \cite{YFK}.  The
new bound is given in terms of Wigner-Yanase-Dyson informations. 
Kosaki himself asked if this inequality can be proved in the setting
of von Neumann algebras.  In this paper we provide a positive answer
to that question and moreover we show how the inequality can be
generalized to an arbitrary operator monotone function.

\smallskip

\noindent 2000 {\sl Mathematics Subject Classification.} Primary
62B10, 94A17; Secondary 46L30, 46L60.

\noindent {\sl Key words and phrases.} Uncertainty principle,
Wigner-Yanase-Dyson information, operator monotone functions.
\end{abstract}

% this part should work with \documentclass{ws-ijm}, but it doesn't

%\begin{document}

%\markboth{Paolo Gibilisco, Tommaso Isola}
%\title{AN INEQUALITY RELATED TO UNCERTAINTY PRINCIPLE IN VON NEUMANN ALGEBRAS}
%\author{Paolo Gibilisco}
%\address{Dipartimento SEFEMEQ, Facolt\`a di
%Economia, Universit\`a di Roma ``Tor Vergata", Via Columbia 2, 00133
%Rome, Italy.  Email: gibilisco@volterra.uniroma2.it -- URL:
%http://www.economia.uniroma2.it/sefemeq/professori/gibilisco} 
%\author{Tommaso Isola}
%\address{Dipartimento di Matematica, Universit\`a di
%Roma ``Tor Vergata", Via della Ricerca Scientifica, 00133 Rome, Italy. 
%Email: isola@mat.uniroma2.it -- URL:
%http://www.mat.uniroma2.it/$\sim$isola} 
%\maketitle

%\begin{abstract}
%Recently Kosaki proved in \cite{Ko} an inequality for matrices that
%can be seen as a kind of new uncertainty principle.  Independently,
%the same result was proved by Yanagi {\sl et al.} in \cite{YFK}.  The
%new bound is given in terms of Wigner-Yanase-Dyson informations. 
%Kosaki himself asked if this inequality can be proved in the setting
%of von Neumann algebras.  In this paper we provide a positive answer
%to that question and moreover we show how the inequality can be
%generalized to an arbitrary operator monotone function.
%\end{abstract}

%\keywords{Uncertainty principle,
%Wigner-Yanase-Dyson information, operator monotone functions.}

%\ccode{Mathematics Subject Classification 2000: Primary
%62B10, 94A17; Secondary 46L30, 46L60.}

\section{Introduction}

If $A,B$ are selfadjoint matrices and $\rho$ is a density matrix, define
\begin{align*}
    \Cov_{\rho}(A,B) &:= \re \{{\rm Tr}(\rho A B)-{\rm Tr}(\rho
    A)\cdot{\rm Tr}(\rho B) \} \\
    \Var_{\rho}(A) &:=\Cov_{\rho}(A,A). 
\end{align*}
The  uncertainty principle reads as 
$$
{\rm Var}_{\rho}(A){\rm Var}_{\rho}(B)
\geq
\frac{1}{4}\vert {\rm Tr}(\rho[A,B])\vert^2.
$$
This inequality can be refined as 
$$
{\rm Var}_{\rho}(A){\rm Var}_{\rho}(B)-\Cov_{\rho}(A,B)^2 \geq
\frac{1}{4}\vert {\rm Tr}(\rho[A,B])\vert^2,
$$
(see \cite{Heis,Schr}).
Recently a different uncertainty principle has been found 
\cite{LZ,LQa,LQb,Ko,YFK}. 
For  $\beta\in(0,1)$ define $\beta$-correlation and $\beta$-information as
\begin{align*}
    \Corr_{\rho,\beta}(A,B) &:=\re\{{\rm Tr}(\rho A B)-{\rm Tr}(\rho^{\beta}
    A\rho^{1-\beta}B) \} \\
    I_{\rho,\beta}(A) &:= \Corr_{\rho,\beta}(A,A) = \Tr(\rho A^{2}) -
    \Tr(\rho^{\beta}A\rho^{1-\beta}A),
\end{align*}
where the latter coincides with the Wigner-Yanase-Dyson information. 
It has been proved that

\begin{equation} \label{betaup}
    \Var_{\rho}(A)\Var_{\rho}(B) - \Cov_{\rho}(A,B)^{2} \geq 
    I_{\rho,\beta}(A)I_{\rho,\beta}(B) - \Corr_{\rho,\beta}(A,B)^{2}. 
\end{equation}

The quantities involved in the previous inequality make a perfect
sense in a von Neumann algebra setting (see for example
\cite{Kosaki:1982/83}).  In ref.  \cite{Ko} Kosaki asked if the
inequality (\ref{betaup}) is true in this more general setting.

In this paper we provide a positive answer to Kosaki question and
moreover we show that, once the inequality is formulated in the
context of operator monotone functions, the result can be greatly
generalized.

\section{Preliminaries}

Denote by $M_{n, sa}$ the space of complex self-adjoint $n \times n$
matrices, and recall that a function $f:(0,\infty)\to {\mathbb R}$ is said
{\it operator monotone} if, for any $n\in {\mathbb N}$, any $A, B\in M_{n,sa}$
such that $0\leq A\leq B$, the inequalities $0\leq f(A)\leq f(B)$
hold.  Then, $f:(0,\infty)\to {\mathbb R}$ is operator monotone {\it iff} for
any $A,B\in{\mathcal B}({\mathcal H})$ such that $0\leq A\leq B$, it holds $f(A)\leq
f(B)$.  An operator monotone function is said {\it symmetric} if
$f(x):=xf(x^{-1})$ and {\it normalized} if $f(1)=1$.  We denote by
${\mathfrak F}$ the class of positive, symmetric, normalized, operator
monotone functions.

Examples of operator monotone functions are the so-called
Wigner-Yanase-Dyson functions
$$
f_{\beta}(x) := \beta (1- \beta)
\frac{(x-1)^2}{(x^{\beta}-1) (x^{1-\beta}-1)}, \qquad \beta\in(0,1). 
$$

Returning to a general $f\in{\mathfrak F}$, we associate to it a 
function $\tilde f\in{\mathfrak F}$ \cite{GII01} defined by 
$$
\tilde{f}(x):=\frac{1}{2}\Bigl( (x+1)-(x-1)^2 \frac{f(0)}{f(x)}
\Bigr),\ x>0.
$$ 

For example
$$
{\tilde f}_{\beta}(x)=\frac{1}{2} (x^{\beta}+x^{1-\beta}).
$$

\begin{Dfn} \label{Dfn:f-correlation}
    For $A,B \in M_{n,sa}$, $f\in{\mathfrak F}$, and $\rho$ a faithful
    density matrix, define $f$-correlation and $f$-information as
    \begin{align*}
	\Corr^{f}_{\rho}(A,B) & := \re \{{\rm Tr}(\rho A B)-{\rm
	Tr}(R_{\rho}\tilde{f}(L_{\rho}R_{\rho}^{-1})(A)\cdot B) \}, \\
	I^{f}_{\rho}(A) & := \Corr^{f}_{\rho}(A,A).
    \end{align*}    
\end{Dfn}

Recall that $f$-information is also known as metric adjusted skew
information (see \cite{Hansen:2006b}).  The following generalization
of inequality (\ref{betaup}) is proved in \cite{GII01}.

\begin{Thm} \label{Thm:GII.Ineq} 
    $$
    \Var_{\rho}(A)\Var_{\rho}(B) - \Cov_{\rho}(A,B)^{2} \geq 
    I^{f}_{\rho}(A)I^{f}_{\rho}(B) - \Corr^{f}_{\rho}(A,B)^{2}.
    $$
\end{Thm}

In the next Section we prove that the above inequality holds true in a
general von Neumann algebra, thus answering, in particular, the
question raised by Kosaki in \cite{Ko}, and recalled above.  A
different generalization of Theorem \ref{Thm:GII.Ineq} has been proved
in \cite{GII04}.

\section{The main result} 
\label{sec:main}

Let ${\mathcal M}$ be a von Neumann algebra, and $\omega$ a normal faithful state
on ${\mathcal M}$, and denote by ${\mathcal H}_{\omega}$ and $\xi_{\omega}$ the GNS Hilbert
space and vector, and by $S_{\omega}$, $J_{\omega}$ and $\Delta_{\omega}$ the modular
operators associated to $\omega$.

The proof of the main result is divided in a series of Lemmas.  In
order to deal with unbounded operators, we introduce some sesquilinear
forms on ${\mathcal H}_{\omega}$, and take \cite{Kato} as our standard reference.

\begin{Dfn}
    Let $f\in{\mathfrak F}$, and define the following sequilinear forms 
    \begin{align*}
	{\mathcal E}(\xi,\eta) & := \langle \Delta_{\omega}^{1/2}\xi, \Delta_{\omega}^{1/2}\eta
	\rangle, \\
	{\mathcal E}_{1}(\xi,\eta) & := {\mathcal E}(\xi,\eta) + \langle \xi,\eta
	\rangle, \\
	{\mathcal F}^f(\xi,\eta) & := \langle \tilde{f}(\Delta_{\omega})^{1/2}\xi,
	\tilde{f}(\Delta_{\omega})^{1/2}\eta \rangle,\\
    	{\mathcal G}^f(\xi,\eta) & := \frac12 {\mathcal E}_1(\xi,\eta) - {\mathcal F}^f(\xi,\eta).
    \end{align*}
\end{Dfn}

It follows from \cite{Kato}, Example VI.1.13, that ${\mathcal E}$, ${\mathcal E}_1$,
${\mathcal F}^f$ are closed, positive and symmetric sesquilinear forms.

\begin{Lemma} \label{Lem:Approx}
    Let $\xi,\eta\in{\mathcal D}(\Delta_{\omega}^{1/2})$, and $\{ \xi_{n} \}$,
    $\{ \eta_{n} \}\subset {\mathcal D}(\Delta_{\omega})$ be such that $\xi_{n}\to\xi$,
    ${\mathcal E}(\xi_{n}-\xi,\xi_{n}-\xi) \to 0$, $n\to\infty$, and
    analogously for $\eta_{n}$ and $\eta$.  Then
    \begin{align*}
	{\mathcal E}(\xi,\eta) & = \lim_{n\to\infty} {\mathcal E}(\xi_{n},\eta_{n}) =
	\lim_{n\to\infty} \langle \xi_{n}, \Delta_{\omega} \eta_{n} \rangle, \\
	{\mathcal F}^f(\xi,\eta) & = \lim_{n\to\infty}
	{\mathcal F}^f(\xi_{n},\eta_{n}) = \lim_{n\to\infty} \langle \xi_{n},
	\tilde{f}(\Delta_{\omega}) \eta_{n} \rangle.
    \end{align*}
\end{Lemma}
\begin{proof}
    It follows from \cite{Kato} Theorem VI.2.1 that ${\mathcal D}(\Delta_{\omega})$ is
    a core for ${\mathcal D}({\mathcal E})\equiv {\mathcal D}(\Delta_{\omega}^{1/2})$, so that, from
    \cite{Kato} Theorem VI.1.21, for any $\xi\in{\mathcal D}(\Delta_{\omega}^{1/2})$
    there is $\{ \xi_{n} \}\subset{\mathcal D}(\Delta_{\omega})$ such that
    $\xi_{n}\to\xi$, and ${\mathcal E}(\xi_{n}-\xi,\xi_{n}-\xi) \to 0$,
    $n\to\infty$.  Then ${\mathcal E}(\xi_{n}-\xi_{m},\xi_{n}-\xi_{m}) \to 0$,
    $m,n\to\infty$.  Now observe that $0 \leq \tilde{f}(x) \leq
    \frac12(x+1)$, for $x>0$ \cite{GII01}, so that
    \begin{align*}
	{\mathcal F}^f(\xi_{n}-\xi_{m},\xi_{n}-\xi_{m}) & = \langle
	\tilde{f}(\Delta_{\omega})^{1/2}(\xi_{n}-\xi_{m}),
	\tilde{f}(\Delta_{\omega})^{1/2}(\xi_{n}-\xi_{m}) \rangle \\
	& = \langle \xi_{n}-\xi_{m},
	\tilde{f}(\Delta_{\omega})(\xi_{n}-\xi_{m}) \rangle \\
	& \leq \frac12 \langle \xi_{n}-\xi_{m}, \xi_{n}-\xi_{m}
	\rangle + \frac12 \langle \xi_{n}-\xi_{m},
	\Delta_{\omega}(\xi_{n}-\xi_{m}) \rangle \\
	& = \frac12 \| \xi_{n}-\xi_{m} \| + \frac12 {\mathcal E}(
	\xi_{n}-\xi_{m}, \xi_{n}-\xi_{m}) \to 0,\ m,n\to\infty.
    \end{align*}
    This implies $\xi\in{\mathcal D}({\mathcal F}^f)$ and
    ${\mathcal F}^f(\xi_{n}-\xi,\xi_{n}-\xi) \to 0$, $n\to\infty$.
    
    Therefore, if $\xi,\eta\in{\mathcal D}(\Delta_{\omega}^{1/2})$, and
    $\{ \xi_{n} \}$, $\{ \eta_{n} \} \subset {\mathcal D}(\Delta_{\omega})$ approximate
    $\xi,\eta$ in the above sense, we obtain, from \cite{Kato} Theorem
    VI.1.12, that ${\mathcal F}^f(\xi,\eta) = \lim_{n\to\infty}
    {\mathcal F}^f(\xi_{n},\eta_{n})$, and analogously for ${\mathcal E}$.
\end{proof}

\begin{Lemma} \label{Lem:FandG}
    \item{$(i)$} ${\mathcal D}({\mathcal F}^f) \supset {\mathcal D}(\Delta_\omega^{1/2})$,
    
    \item{$(ii)$} ${\mathcal G}^f$ is a symmetric sesquilinear form on ${\mathcal D}({\mathcal G}^f)
    \supset {\mathcal D}(\Delta_\omega^{1/2})$, which is positive on
    ${\mathcal D}(\Delta_\omega^{1/2})$.
\end{Lemma}
\begin{proof}
    $(i)$ It follows from the proof of the previous Lemma.
	
    \noindent $(ii)$ We only need to prove positivity.  To begin with, let
    $\xi\in {\mathcal D}(\Delta_\omega)$.  Then, setting $g(x):=
    \frac12(x+1)-\tilde{f}(x) \geq 0$, for all $x>0$, we have
    ${\mathcal G}^f(\xi,\xi) = \frac12 {\mathcal E}_1(\xi,\xi) - {\mathcal F}^f(\xi,\xi) =
    \frac12 \langle \xi, \xi \rangle + \frac12 \langle \xi, \Delta_{\omega}\xi
    \rangle -\langle \xi, \tilde{f}(\Delta_{\omega})\xi \rangle = \langle \xi,
    g(\Delta_{\omega})\xi \rangle \geq 0$.
    
    Moreover, if $\xi\in {\mathcal D}(\Delta_\omega^{1/2})$, and $\xi_n\in {\mathcal D}(\Delta_\omega)$
    is such that $\xi_n\to \xi$, and ${\mathcal E}(\xi_n-\xi,\xi_n-\xi)\to 0$,
    then, from Lemma \ref{Lem:Approx} it follows ${\mathcal G}^f(\xi,\xi) =
    \lim_{n\to\infty} {\mathcal G}^f(\xi_n,\xi_n) \geq 0$.
\end{proof}

We can now introduce the main objects of study.  In the sequel, we
denote by $T\widehat\in{\mathcal M}$ the fact that $T$ is a closed, densely
defined, linear operator on ${\mathcal H}_\omega$, and is affiliated with ${\mathcal M}$.

\begin{Dfn}
    For any $A,B\widehat\in{\mathcal M}_{sa}$, such that $\xi_{\omega}\in
    {\mathcal D}(A)\cap{\mathcal D}(B)$, and any $f\in{\mathfrak F}$, we set
    $A_{0}:=A-\langle \xi_{\omega}, A\xi_{\omega} \rangle$, $B_{0}:=B-\langle
    \xi_{\omega}, B\xi_{\omega} \rangle$, and define the bilinear forms
    \begin{align*}	    
	\Cov_{\omega}(A,B) & := \re \langle
	A_{0}\xi_{\omega}, B_{0}\xi_{\omega} \rangle, \\
	\Var_{\omega}(A) & := \Cov_{\omega}(A,A),\\
	\Corr^{f}_{\omega}(A,B) & := \re \langle A_{0}\xi_{\omega}, B_{0}\xi_{\omega}
	\rangle - \re \langle \tilde{f}(\Delta_{\omega})^{1/2} A_{0}\xi_{\omega},
	\tilde{f}(\Delta_{\omega})^{1/2}B_{0}\xi_{\omega} \rangle, \\
	I^{f}_{\omega}(A) & := \Corr^{f}_{\omega}(A,A).
    \end{align*}
\end{Dfn}

\begin{Rem}
    Observe that in the matrix case  $\omega=\Tr(\rho
    \cdot)$, for some density matrix $\rho$, and
    $\Delta_{\omega}=L_{\rho}R_{\rho}^{-1}$, so that the previous Definition is a
    true generalization of covariance and $f$-correlation in the
    matrix case.
\end{Rem}

For the reader's convenience, we prove the following folklore
result.

\begin{Lemma} \label{Lem:DomainOfS}
    ${\mathcal D}(\Delta_{\omega}^{1/2}) = \{ T\xi_{\omega}: T\widehat\in{\mathcal M},
    \xi_{\omega}\in{\mathcal D}(T)\cap{\mathcal D}(T^{*}) \}$.
\end{Lemma}
\begin{proof}
    $(1)$ Let us first prove that ${\mathcal D}(\Delta_{\omega}^{1/2}) \subset
    \{ T\xi_{\omega}: T\widehat\in{\mathcal M},
    \xi_{\omega}\in{\mathcal D}(T)\cap{\mathcal D}(T^{*}) \}$.  Indeed, let $\eta\in
    {\mathcal D}(\Delta_\omega^{1/2})$, and define the linear operator $T_0 : x'\xi_\omega
    \in {\mathcal M}'\xi_\omega \mapsto x'\eta \in {\mathcal H}_\omega$, which is densely
    defined, and affiliated with ${\mathcal M}$.  Let us show that is
    preclosed: indeed, if $x'_n\xi_{\omega} \to 0$, and $x'_n\eta \to \zeta$,
    then, for any $y'\in{\mathcal M}'$, we get
    \begin{align*}
	\langle \zeta, y'\xi_\omega \rangle & = \lim_{n\to\infty} \langle
	x'_n\eta,y'\xi_\omega \rangle = \lim_{n\to\infty} \langle
	\eta,{x'_n}^*y'\xi_\omega \rangle = \lim_{n\to\infty} \langle
	\eta, S_{\omega}^* ({y'}^*{x'_n}\xi_\omega) \rangle \\
	& = \lim_{n\to\infty} \langle {y'}^*{x'_n}\xi_\omega,S_{\omega}\eta
	\rangle = \lim_{n\to\infty} \langle {x'_n}\xi_\omega, y'S_{\omega}\eta
	\rangle = 0,
    \end{align*}	
    which shows that $T_0$ is preclosed.  Let $T_\eta := \overline{T_0}$. 
    Then, $T_\eta\widehat\in{\mathcal M}$, and $T_\eta \xi_\omega = \eta$.  It
    remains to be proved that $\xi_\omega\in{\mathcal D}(T_\eta^*)$.  Since
    $S_{\omega}\eta\in{\mathcal D}(\Delta_{\omega}^{1/2})$, we can also consider
    $T_{S_{\omega}\eta}$.  Let us show that $T_{S_{\omega}\eta} \subset
    T_{\eta}^{*}$.  Indeed, for any $x',y'\in{\mathcal M}'$, we have
    $$
    \langle T_{S_{\omega}\eta} x'\xi_{\omega}, y'\xi_\omega \rangle = \langle
    x'S_{\omega}\eta, y'\xi_\omega \rangle = \langle S_{\omega}\eta,
    {x'}^{*}y'\xi_\omega \rangle = \langle {y'}^{*}x'\xi_\omega ,\eta \rangle
    = \langle x'\xi_\omega ,y'\eta \rangle = \langle x'\xi_\omega
    ,T_{\eta}y'\xi_{\omega} \rangle.
    $$
    Then, $\xi_{\omega}\in{\mathcal D}(T_{S_{\omega}\eta}) \subset {\mathcal D}(T_{\eta}^{*})$,
    which shows that ${\mathcal D}(\Delta_{\omega}^{1/2}) \subset \{ T\xi_{\omega}:
    T\widehat\in{\mathcal M}, \xi_{\omega}\in{\mathcal D}(T)\cap{\mathcal D}(T^{*}) \}$.
    
    \noindent $(2)$ Let us now prove that ${\mathcal D}(\Delta_{\omega}^{1/2}) \supset
    \{ T\xi_{\omega}: T\widehat\in{\mathcal M},
    \xi_{\omega}\in{\mathcal D}(T)\cap{\mathcal D}(T^{*}) \}$.  Indeed, if $T\widehat\in{\mathcal M}$
    is such that $\xi_{\omega}\in{\mathcal D}(T)\cap{\mathcal D}(T^{*})$, we can consider
    its polar decomposition $T=v|T|$, and let $e_{n}:=
    \chi_{[0,n]}(|T|)$, $T_{n} := v|T|e_{n}$, for any $n\in{\mathbb N}$. 
    Since $\xi_{\omega}\in{\mathcal D}(T)$, we have $T_{n}\xi_{\omega} =
    ve_{n}|T|\xi_{\omega} \to T\xi_{\omega}$.  Moreover, since
    $\xi_{\omega}\in{\mathcal D}(T^{*})$, we have $T_{n}^{*}\xi_{\omega} =
    |T|e_{n}v^{*}\xi_{\omega} = e_{n}T^{*}\xi_{\omega} \to T^{*}\xi_{\omega}$. 
    Since $S_{\omega}$ is a closed operator, it follows that
    $T\xi_{\omega}\in{\mathcal D}(S_{\omega})={\mathcal D}(\Delta_{\omega}^{1/2})$ [and
    $S_{\omega}T\xi_{\omega}=T^{*}\xi_{\omega}$], which is what we wanted to
    prove.
\end{proof}

\begin{Lemma} \label{Lem:sesqui}

    For any $A,B\widehat\in{\mathcal M}_{sa}$, such that $\xi_{\omega}\in
    {\mathcal D}(A)\cap{\mathcal D}(B)$, and any $f\in{\mathfrak F}$, we have

    \item{$(i)$} $\Cov_{\omega}(A,B) = \frac12 \re
    {\mathcal E}_{1}(A_{0}\xi_{\omega},B_{0}\xi_{\omega})$ is a positive bilinear form,
    
    \item{$(ii)$} $\Corr^{f}_{\omega}(A,B) = \re
    {\mathcal G}^f(A_{0}\xi_{\omega},B_{0}\xi_{\omega})$ is a positive bilinear form.
\end{Lemma}
\begin{proof}
    $(i)$ Observe that 
    \begin{align*}
	\langle B_{0}\xi_{\omega}, A_{0}\xi_{\omega} \rangle & = \langle
	B_{0}^{*}\xi_{\omega}, A_{0}^{*}\xi_{\omega} \rangle = \langle
	J_{\omega}\Delta_{\omega}^{1/2} B_{0}\xi_{\omega}, J_{\omega}\Delta_{\omega}^{1/2}
	A_{0}\xi_{\omega} \rangle \\
	& = \langle \Delta_{\omega}^{1/2} A_{0}\xi_{\omega}, \Delta_{\omega}^{1/2}
	B_{0}\xi_{\omega} \rangle = {\mathcal E}(A_{0}\xi_{\omega},B_{0}\xi_{\omega}).
    \end{align*}
    The thesis follows from this and the fact that ${\mathcal D}(\Delta_{\omega}^{1/2})
    = \{ T\xi_{\omega}: T\widehat\in{\mathcal M},
    \xi_{\omega}\in{\mathcal D}(T)\cap{\mathcal D}(T^{*}) \}$.
    
    \noindent $(ii)$ It follows from $(i)$ and Lemma \ref{Lem:FandG} $(ii)$.
\end{proof}

\begin{Lemma} \label{Lem:PosMeasure}
    Let $\xi,\eta\in{\mathcal H}_\omega$, $\Delta_{\omega} = \int_{0}^{\infty} t\, de(t)$,
    and define, for $\Omega$ a Borel subset of $[0,\infty)$,
    $\mu_{\xi\eta}(\Omega):= \re \langle \xi, e(\Omega)\eta \rangle$, and
    $$
    \mu:= \mu_{\xi\xi}\otimes\mu_{\eta\eta} +
    \mu_{\eta\eta}\otimes\mu_{\xi\xi} - 2
    \mu_{\xi\eta}\otimes\mu_{\xi\eta}.
    $$
    Then, $\mu$ is a bounded positive Borel measure on
    $[0,\infty)^{2}$.
\end{Lemma}
\begin{proof}
    Let $\Omega_1,\Omega_2$ be Borel subsets of $[0,\infty)$, and set
    $e_{j}:=e(\Omega_{j})$, $j=1,2$.  Observe that $| \re \langle \xi,
    e_{1}\eta \rangle \cdot \re \langle \xi, e_{2}\eta \rangle | \leq
    \|e_{1}\xi\| \cdot \|e_{1}\eta\| \cdot \|e_{2}\xi\| \cdot
    \|e_{2}\eta\|$, so that
    $$
    \mu(\Omega_{1}\times \Omega_{2}) \geq \| e_{1}\xi \|^{2} \cdot \| e_{2}\eta
    \|^{2}+\| e_{2}\xi \|^{2} \cdot \| e_{1}\eta \|^{2} - 2 \|
    e_{1}\xi \| \cdot \| e_{1}\eta \| \cdot \| e_{2}\xi \| \cdot \|
    e_{2}\eta \| \geq 0.
    $$
    The thesis follows by standard measure theoretic arguments.
\end{proof}

\begin{Thm}
   For any $A,B\widehat\in{\mathcal M}_{sa}$, such that $\xi_{\omega}\in 
   {\mathcal D}(A)\cap{\mathcal D}(B)$, and any $f\in{\mathfrak F}$, we have
   $$
   \Var_{\omega}(A)\Var_{\omega}(B)- \Cov_{\omega}(A,B)^{2} \geq
   I^{f}_{\omega}(A)I^{f}_{\omega}(B) - \Corr^{f}_{\omega}(A,B)^{2}.
   $$
\end{Thm}
\begin{proof}
    Set 
    \begin{align*}
	G(A,B) & := \Var_{\omega}(A)\Var_{\omega}(B)- \Cov_{\omega}(A,B)^{2} -
	I^{f}_{\omega}(A)I^{f}_{\omega}(B) + \Corr^{f}_{\omega}(A,B)^{2} \\
	& \stackrel{(a)}{=} \frac12 {\mathcal E}_{1}(A_{0}\xi_{\omega},A_{0}\xi_{\omega}) \cdot \frac12
	{\mathcal E}_{1}(B_{0}\xi_{\omega},B_{0}\xi_{\omega}) - \Bigl( \frac12 \re
	{\mathcal E}_{1}(A_{0}\xi_{\omega},B_{0}\xi_{\omega}) \Bigr)^{2} \\
	& \quad - \Bigl( \frac12 {\mathcal E}_1(A_{0}\xi_{\omega},A_{0}\xi_{\omega}) -
	{\mathcal F}^f(A_{0}\xi_{\omega},A_{0}\xi_{\omega}) \Bigr) \Bigl( \frac12
	{\mathcal E}_1(B_{0}\xi_{\omega},B_{0}\xi_{\omega}) -
	{\mathcal F}^f(B_{0}\xi_{\omega},B_{0}\xi_{\omega}) \Bigr) \\
	& \quad + \Bigl( \frac12 \re
	{\mathcal E}_1(A_{0}\xi_{\omega},B_{0}\xi_{\omega}) - \re
	{\mathcal F}^f(A_{0}\xi_{\omega},B_{0}\xi_{\omega}) \Bigr)^{2} \\
	& = \frac12 {\mathcal E}_{1}(A_{0}\xi_{\omega},A_{0}\xi_{\omega}) \cdot
	{\mathcal F}^f(B_{0}\xi_{\omega},B_{0}\xi_{\omega}) + \frac12
	{\mathcal F}^f(A_{0}\xi_{\omega},A_{0}\xi_{\omega}) \cdot
	{\mathcal E}_{1}(B_{0}\xi_{\omega},B_{0}\xi_{\omega}) \\
	& \quad - {\mathcal F}^f(A_{0}\xi_{\omega},A_{0}\xi_{\omega})\cdot
	{\mathcal F}^f(B_{0}\xi_{\omega},B_{0}\xi_{\omega}) - \re
	{\mathcal E}_{1}(A_{0}\xi_{\omega},B_{0}\xi_{\omega})\cdot
	\re{\mathcal F}^f(A_{0}\xi_{\omega},B_{0}\xi_{\omega}) \\
	& \quad + \bigl( \re{\mathcal F}^f(A_{0}\xi_{\omega},B_{0}\xi_{\omega})
	\bigr)^{2},
    \end{align*}
    where in $(a)$ we have used Lemma \ref{Lem:sesqui}.
    Let us now introduce the function, for
    $\xi,\eta\in{\mathcal D}(\Delta_{\omega}^{1/2})$,
    $$
    H(\xi,\eta) := \frac12 {\mathcal E}_{1}(\xi,\xi) \cdot {\mathcal F}^f(\eta,\eta) +
    \frac12 {\mathcal F}^f(\xi,\xi) \cdot {\mathcal E}_{1}(\eta,\eta) - {\mathcal F}^f(\xi,\xi)\cdot
    {\mathcal F}^f(\eta,\eta) - \re {\mathcal E}_{1}(\xi,\eta)\cdot \re{\mathcal F}^f(\xi,\eta) +
    \bigl( \re{\mathcal F}^f(\xi,\eta) \bigr)^{2},
    $$
    and recall that ${\mathcal D}(\Delta_{\omega}^{1/2}) = \{ T\xi_{\omega}:
    T\widehat\in{\mathcal M}, \xi_{\omega}\in{\mathcal D}(T)\cap{\mathcal D}(T^{*}) \}$, so that, if
    $A,B$ are as in the statement of the Theorem, we obtain $G(A,B) =
    H(A_{0}\xi_{\omega},B_{0}\xi_{\omega})$, and to prove the theorem it
    suffices to show that $H(\xi,\eta)\geq 0$, for all
    $\xi,\eta\in{\mathcal D}(\Delta_\omega^{1/2})$.  Observe that, for
    $\xi,\eta\in{\mathcal D}(\Delta_{\omega})$, we get
    \begin{align*}
	H(\xi,\eta) & = \frac12 \langle \xi, (1+\Delta_{\omega}) \xi \rangle
	\cdot \langle \eta, \tilde{f}(\Delta_{\omega}) \eta \rangle + \frac12
	\langle \eta, (1+\Delta_{\omega}) \eta \rangle \cdot \langle \xi,
	\tilde{f}(\Delta_{\omega}) \xi \rangle \\
	& \quad - \langle \xi, \tilde{f}(\Delta_{\omega}) \xi \rangle \cdot
	\langle \eta, \tilde{f}(\Delta_{\omega}) \eta \rangle - \re \langle
	\xi, (1+\Delta_{\omega}) \eta \rangle \cdot \re \langle \xi,
	\tilde{f}(\Delta_{\omega}) \eta \rangle + \bigl( \re \langle \xi,
	\tilde{f}(\Delta_{\omega}) \eta \rangle \bigr)^{2}\\
	& \stackrel{(b)}{=} \frac12 \int_{0}^{\infty} (s+1)\,
	d\mu_{\xi\xi}(s) \int_{0}^{\infty} \tilde{f}(t)\,
	d\mu_{\eta\eta}(t) + \frac12 \int_{0}^{\infty} \tilde{f}(s) \,
	d\mu_{\xi\xi}(s) \int_{0}^{\infty} (t+1) \, d\mu_{\eta\eta}(t)
	\\
	& \quad - \int_{0}^{\infty} \tilde{f}(s) \, d\mu_{\xi\xi}(s)
	\int_{0}^{\infty} \tilde{f}(t) \, d\mu_{\eta\eta}(t) - \frac12
	\int_{0}^{\infty} (s+1)\, d\mu_{\xi\eta}(s) \int_{0}^{\infty}
	\tilde{f}(t)\, d\mu_{\xi\eta}(t) \\
	& \quad - \frac12 \int_{0}^{\infty} \tilde{f}(s) \,
	d\mu_{\xi\eta}(s) \int_{0}^{\infty} (t+1) \, d\mu_{\xi\eta}(t)
	- \int_{0}^{\infty} \tilde{f}(s) \, d\mu_{\xi\eta}(s)
	\int_{0}^{\infty} \tilde{f}(t) \, d\mu_{\xi\eta}(t) \\
	& \stackrel{(c)}{=} \frac12 \int_{[0,\infty)^{2}} \bigl(
	(s+1)\tilde{f}(t) +(t+1)\tilde{f}(s) -
	2\tilde{f}(s)\tilde{f}(t) \bigr) \, d\mu_{\xi\xi}\otimes
	\mu_{\eta\eta}(s,t) \\
	& - \frac12 \int_{[0,\infty)^{2}} \bigl(
	(s+1)\tilde{f}(t) +(t+1)\tilde{f}(s) -
	2\tilde{f}(s)\tilde{f}(t) \bigr) \, d\mu_{\xi\eta}\otimes
	\mu_{\xi\eta}(s,t) \\
	& \stackrel{(d)}{=} \frac14 \iint_{[0,\infty)^{2}} \bigl(
	(s+1)\tilde{f}(t) + (t+1)\tilde{f}(s)
	-2\tilde{f}(s)\tilde{f}(t) \bigr) \, d\mu(s,t),
    \end{align*}
    where we used in $(b)$ notation as in Lemma \ref{Lem:PosMeasure},
    in $(c)$ Fubini-Tonelli Theorem, and in $(d)$ the symmetries of
    the first integrand and notation as in Lemma \ref{Lem:PosMeasure}. 
    
    \noindent Since $\mu$ is a positive measure, and
    $$
    (s+1)\tilde{f}(t) + (t+1)\tilde{f}(s) -2\tilde{f}(s)\tilde{f}(t)
    = \bigl(s+1- \tilde{f}(s)\bigr)\tilde{f}(t) +
    \bigl(t+1-\tilde{f}(t) \bigr) \tilde{f}(s) \geq 0,
    $$
    we obtain $H(\xi,\eta)\geq 0$, for any $\xi,\eta\in{\mathcal D}(\Delta_{\omega})$.
        
    It follows from Lemma \ref{Lem:Approx} that, for any
    $\xi,\eta\in{\mathcal D}(\Delta_{\omega}^{1/2})$, we have $H(\xi,\eta) =
    \lim_{n\to\infty} H(\xi_{n},\eta_{n})\geq 0$, which ends the
    proof.
\end{proof}

\end{document}